\documentclass[conference]{IEEEtran}
\IEEEoverridecommandlockouts

\usepackage{amsmath,amssymb,amsfonts,amsthm}
\usepackage{graphicx}
\usepackage{textcomp}
\usepackage{xcolor}
\usepackage{hyperref}
\usepackage{array}
\usepackage{booktabs}
\usepackage{bm}
\usepackage{tabularx}
\usepackage{tikz}
\usepackage[letterpaper,top=0.8in,bottom=1.09in,left=0.66in,right=0.66in]{geometry}

\usetikzlibrary{positioning,arrows.meta}

\def\BibTeX{{\rm B\kern-.05em{\sc i\kern-.025em b}\kern-.08em
    T\kern-.1667em\lower.7ex\hbox{E}\kern-.125emX}}

\newtheorem{theorem}{Theorem}
\newtheorem{lemma}{Lemma}
\newtheorem{corollary}{Corollary}
\newtheorem{remark}{Remark}
\newtheorem{definition}{Definition}
\newtheorem{example}{Example}
\newtheorem{proposition}{Proposition}

\begin{document}

\title{Combinatorial Capacity Bounds for the $q$-ary Deletion Channel $^{\star}$}

\author{
  \IEEEauthorblockN{Hassan Tavakoli}
  \IEEEauthorblockA{School of EECS\\
                    Oregon State University\\
                    Oregon, OR 97331, USA\\
                    tavakolh@oregonstate.edu}
  \and
  \IEEEauthorblockN{Thinh Nguyen, \emph{Senior Member, IEEE}}
  \IEEEauthorblockA{School of EECS\\
                    Oregon State University\\
                    Oregon, OR 97331, USA\\
                    thinhq@eecs.oregonstate.edu}
  \and
  \IEEEauthorblockN{Bella Bose, \emph{Life Fellow, IEEE}\thanks{ $^{\star}$ This This work was supported by the National Science Foundation under Grant No. CCF:SHF:2417898.}}
  \IEEEauthorblockA{School of EECS\\
                    Oregon State University\\
                    Oregon, OR 97331, USA\\
                    Bella.Bose@oregonstate.edu}
}

\maketitle

\begin{abstract}
We study the \(q\)-ary deletion channel via the pattern-count scalar
\(N_n(x,y)\), the number of deletion subsets mapping \(x\in\Sigma_q^n\)
to \(y\in\Sigma_q^k\), which factorizes the transition probability.
Two sum identities on \(N_n\) certify stochastic normalization and,
under uniform input, yield an exact closed-form output entropy.
These give the finite-block capacity sandwich
\(
(1-d)\log_2 q-h_2(d)\;\le\; C_{q,n}\;\le\;(1-d)\log_2 q.
\)
The exact uniform-input rate is
\(
\frac{1}{n}I_U(X;Y)
=(1-d)\log_2 q+\frac{1}{n}H_{\mathrm{Bin}}(n,1-d)-h_2(d)+\frac{\Delta_n(d)}{n},
\)
from which the simpler certified bound
\(
C_{q,n}\ge (1-d)\log_2 q-h_2(d)+\frac{\Delta_n(d)}{n}
\)
follows.
The small-\(d\) bound \(C_q(d)\ge\log_2 q+d\log_2 d+O(d)\) follows
for all \(q\ge 2\). Numerical experiments at \(n=3,5,10\) and \(q=2,3\)
confirm all bounds.
\end{abstract}


\section{Introduction}
\label{sec:intro}

The deletion channel erases each transmitted symbol independently with
probability $d$ and delivers the surviving subsequence without marking
which positions were deleted. Despite decades of study, tight capacity
characterizations remain open~\cite{mitzenmacher2009}.
The key object in this paper is the integer scalar
\begin{equation}
  N_n(x,y)\triangleq \#\bigl\{S\subseteq[n]:|S|=n-k,\,x_{\setminus S}=y\bigr\},
  \label{eq:Nn}
\end{equation}
where $x_{\setminus S}$ is the word obtained by deleting the positions in $S$ from $x$. 
~\eqref{eq:Nn} counts deletion subsets mapping $x\in\Sigma_q^n$ to
$y\in\Sigma_q^k$ and factorizes the transition probability as
$\Pr(Y=y\mid X=x)=N_n(x,y)\cdot d^{n-k}(1-d)^k$~\cite{mitzenmacher2009}.
The scalar $N_n(x,y)$ satisfies a two-case recursion uniform in $q\ge 2$,
separating the combinatorial structure of the channel from its
probabilistic weights and organizing all capacity calculations through
a single integer array.

\paragraph*{Prior work}
Capacity lower bounds for the BDC are due to Gallager~\cite{gallager1961}
and Drinea--Mitzenmacher~\cite{drinea2007,mitzenmacher2006}; the
Kanoria--Montanari expansion gives small-$d$ asymptotics for
$q=2$~\cite{kanoria2013}. Upper bounds appear in~\cite{dalai2011,rahmati2014}.
Related capacity bounds for cascaded and related deletion models appear
in~\cite{tavakoli2022}, and improved upper bounds for the binary deletion
channel appear in~\cite{tavakoli2021}.
Fertonani--Duman computed BA capacity numerically for $q=2$,
$n\le 10$~\cite{fertonani2010}. The present work provides a
self-contained $q$-ary treatment via $N_n(x,y)$ with explicit
entropy correction and certified capacity bounds for all $q\ge 2$.

The rest of the paper is organized as follows: Section~\ref{sec:model} introduces $N_n(x,y)$ and its factorization
of the transition probability. Then, gives the
two-case recursion for $N_n$. Section~\ref{sec:bounds} at first proves the
row- and column-sum identities and derives the
exact output entropy under uniform input and the capacity sandwich.
Section~\ref{sec:tighter} introduces the entropy correction $\Phi_{k,n}$,
proves the tightened lower bound with $\Delta_n(d)>0$ certified, and
establishes the small-$d$ bound $C_q(d)\ge\log_2 q+d\log_2 d+O(d)$.
Section~\ref{sec:numerical} presents numerical verification.

\paragraph*{Notation}
$\Sigma_q=\{0,\ldots,q-1\}$; $[n]\triangleq \{1,\ldots,n\}$;
$h_2(p)\triangleq -p\log_2 p-(1-p)\log_2(1-p)$;
$H_{\mathrm{Bin}}(n,p)\triangleq -\sum_{k=0}^n\binom{n}{k}p^k(1-p)^{n-k}
\log_2[\binom{n}{k}p^k(1-p)^{n-k}]$.
Words are ordered lexicographically; $\varepsilon$ denotes the empty
string; $x_{\setminus S}$ denotes the word obtained by deleting
positions $S$ from $x$.  We write $\sigma_{k,n}\triangleq d^{n-k}(1-d)^k$ and
$w_k\triangleq \binom{n}{k}d^{n-k}(1-d)^k$. Also, $\{0,1\}^2$ means $\Sigma^2_2=\{00,01,10,11\}$ and it is extendable to $\Sigma^n_2=\{0,1\}^n$.
Throughout, $H(\cdot)$ denotes entropy and $I(X;Y)$ denotes the mutual information between $X$ and $Y$.
$H_{\mathrm{Bin}}(n,p)$ is also the entropy of a $\mathrm{Bin}(n,p)$ random variable. \(\log\) is in base 2 and \(h(d)\) is
binary entropy on base-2, also, $0\log_2 0\triangleq 0$.

\section{Channel Model and Pattern-Count Scalar}
\label{sec:model}

\subsection{Channel Definition}

Let the input be $X=x_1\cdots x_n\in\Sigma_q^n$.  Each symbol $x_j$
is deleted independently with probability $d\in[0,1)$; surviving
symbols concatenate in order to form $Y\in\bigcup_{k=0}^n\Sigma_q^k$.

\subsection{The Pattern-Count Scalar}

\begin{lemma}[Pattern-count factorization]
\label{lem:factorization}
For all $x\in\Sigma_q^n$ and $y\in\Sigma_q^k$:
\(
  \Pr(Y=y\mid X=x)=N_n(x,y)\cdot\sigma_{k,n}
  =N_n(x,y)\cdot d^{n-k}(1-d)^k.
\)
\end{lemma}

\begin{proof}
Conditioning on the deletion subset $S$: the event that exactly the
positions in $S$ are deleted has probability $d^{|S|}(1-d)^{n-|S|}$.
For output $y$ with $|y|=k$ we need $|S|=n-k$ and $x_{\setminus S}=y$.
Summing over all such $S$ gives
$\Pr(Y=y\mid X=x)=N_n(x,y)\cdot d^{n-k}(1-d)^k$.
\end{proof}

The weight $\sigma_{k,n}$ is the probability that a specific fixed
set of $n-k$ positions are deleted and the $k$ remaining survive.
The integer $N_n(x,y)$ counts how many such sets exist.
The extreme cases are immediate from the definition.  If $k=n$ (no
deletions), the only output is $y=x$, so $N_n(x,x)=1$ and
$N_n(x,y)=0$ for $y\ne x$.  If $k=0$ (all symbols deleted), the
only output is the empty string $\varepsilon$, so $N_n(x,\varepsilon)=1$
for all $x$.

\begin{example}[$n=2$, $q=2$]
\label{ex:n2q2}
For $x=00\in\{0,1\}^2$ and $y=0\in\{0,1\}^1$:
$N_2(00,0)=2$ (delete position 1, or delete position 2; both yield $0$).
For $x=01$ and $y=0$: $N_2(01,0)=1$ (only deleting position 2 yields
$0$; deleting position 1 gives $1$).
The complete scalar arrays for $|y|=k=1$ are:
\begin{table}[h]
\centering
{\arraycolsep=3pt\small$\begin{array}{c|cc} x\backslash y & 0 & 1 \\\hline 00 & 2 & 0 \\ 01 & 1 & 1 \\ 10 & 1 & 1 \\ 11 & 0 & 2 \end{array}$}
\end{table}
\end{example}

\subsection{Recursive Structure of $N_n$}
\label{sec:recursion}

\begin{proposition}[$q$-ary recursion for $N_n$]
\label{thm:recursion}
For $n\ge 1$, $1\le k\le n-1$, and any $q\ge 2$, write $x=sx'$ with
$s\in\Sigma_q$, $x'\in\Sigma_q^{n-1}$, and $y=ty'$ with $t\in\Sigma_q$,
$y'\in\Sigma_q^{k-1}$.  Then:
\begin{equation}
  N_n(sx',ty')=N_{n-1}(x',ty')+\mathbf{1}_{\{s=t\}}N_{n-1}(x',y').
  \label{eq:rec}
\end{equation}
\end{proposition}
\begin{proof}
This is followed by conditioning
on whether the first symbol of $x$ is deleted or survives; an equivalent
identity appears in~\cite{mitzenmacher2009}. The recursion involves only
symbol equality $s=t$ and holds uniformly for all $q\ge 2$.
\end{proof}

\section{Main Results 1: Entropy Bounds and Capacity Sandwich}
\label{sec:bounds}

\subsection{Row-Sum and Column-Sum Identities}

\begin{lemma}[Row-sum identity]
\label{lem:rowsum}
For all $q\ge 2$, and fix \(k\), where $0\le k\le n$, and $x\in\Sigma_q^n$:
\begin{equation}
  \sum_{y\in\Sigma_q^k}N_n(x,y)=\binom{n}{k}.
  \label{eq:rowsum}
\end{equation}
\end{lemma}

\begin{proof}
Fix $x$.  Each $(n-k)$-element subset $S\subseteq[n]$ produces a
unique output $x_{\setminus S}\in\Sigma_q^k$.  There are exactly
$\binom{n}{n-k}=\binom{n}{k}$ such subsets.
\end{proof}

\begin{lemma}[Column-sum identity]
\label{lem:colsum}
For all $q\ge 2$, and fix \(n\), where $0\le k\le n$, $y\in\Sigma_q^k$:
\begin{equation}
  \sum_{x\in\Sigma_q^n}N_n(x,y)=q^{n-k}\binom{n}{k}.
  \label{eq:colsum}
\end{equation}
\end{lemma}

\begin{proof}
Fix $y\in\Sigma_q^k$.  Count pairs $(x,S)$ with $x\in\Sigma_q^n$,
$|S|=n-k$, and $x_{\setminus S}=y$.  For each of the $\binom{n}{k}$
choices of $S$, the $k$ surviving positions of $x$ are fixed to $y$,
while the $n-k$ deleted positions are free over $\Sigma_q$,
contributing $q^{n-k}$ pairs per choice.
Hence $\sum_x N_n(x,y)=q^{n-k}\binom{n}{k}$.
\end{proof}

\subsection{Uniform Output Structure}

\begin{remark}[Why we use the uniform input]
\label{rem:uniform_input}
We use the uniform input distribution here for tractability, not as a
claim of global optimality. For the binary deletion channel, the
i.i.d.\ Bernoulli$(1/2)$ source is asymptotically optimal in the
small-deletion regime $d\to 0^+$, while memoryless sources are not
asymptotically optimal as $d\to 1^-$. \cite{kanoria2013,drmota2012}

In the present paper, uniform input is chosen because it yields the
layerwise-uniform output property in Lemma~\ref{lem:uniform} and the
closed-form entropy formula in Lemma~\ref{thm:HY}. This is an
analytic convenience, not a statement that the uniform input is
capacity-achieving for all $q$ and all $d$.
\end{remark}

\begin{lemma}[Uniform output under uniform input]
\label{lem:uniform}
Under the uniform input $X\sim\mathrm{Unif}(\Sigma_q^n)$, the output
$Y$ conditioned on $|Y|=k$ is uniformly distributed over $\Sigma_q^k$,
with marginal weight $w_k\triangleq \binom{n}{k}d^{n-k}(1-d)^k$.
\end{lemma}

\begin{proof}
For any $y\in\Sigma_q^k$, using Lemma~\ref{lem:factorization}
and the column-sum identity~\eqref{eq:colsum}:
\(
  \Pr(Y=y)
  =q^{-n}\sum_{x\in\Sigma_q^n}N_n(x,y)\,\sigma_{k,n}\\
  =q^{-n}\cdot d^{n-k}(1-d)^k\cdot q^{n-k}\tbinom{n}{k}
  =w_k\cdot q^{-k},
\)
which is independent of $y$.
\end{proof}

\subsection{Exact Output Entropy}

\begin{lemma}[Exact output entropy under uniform input]
\label{thm:HY}
\begin{equation}
  H(Y)=n(1-d)\log_2 q+H_{\mathrm{Bin}}(n,1-d).
  \label{eq:HY}
\end{equation}
\end{lemma}

\begin{proof}
From Lemma~\ref{lem:uniform}, each word of length $k$ has probability
$w_k q^{-k}$.  The contribution to $H(Y)$ from length-$k$ words is
$q^k\cdot\bigl(-\frac{w_k}{q^k}\log_2\frac{w_k}{q^k}\bigr)
=w_k[k\log_2 q-\log_2 w_k]$.
Summing over $k$, $\sum_k w_k=1$ and
$\sum_k kw_k=n(1-d)$:
\(
H(Y)=n(1-d)\log_2 q-\sum_k w_k\log_2 w_k
=n(1-d)\log_2 q+H_{\mathrm{Bin}}(n,1-d).
\)
\end{proof}

\subsection{Finite-Block Capacity Sandwich}

\begin{theorem}[Finite-block sandwich]
\label{thm:sandwich}
For $n\ge 1$, $d\in(0,1)$:
\begin{equation}
  (1-d)\log_2 q-h_2(d) \;\le\;C_{q,n} (d)\;\le\; (1-d)\log_2 q.
  \label{eq:sandwich}
\end{equation}
We denote by $C_{q,n}$ the per-symbol capacity of the $q$-ary 
deletion channel with block length $n$, and by $C_q$ its 
limit as $n\to\infty$, which satisfies
$C_q\le (1-d)\log_2 q$.
\end{theorem}

\begin{proof}
\textit{Lower bound.}
It suffices to lower-bound $\frac{1}{n}I_U(X;Y)$, which means mutual information under uniform input distribution.
By Lemma~\ref{thm:HY}, $H_U(Y)\ge n(1-d)\log_2 q$.
For any fixed $x$, the output $Y$ is a deterministic function of the i.i.d.\ deletion
pattern $\mathbf{D}=(D_1,\ldots,D_n)$, $D_j\sim\mathrm{Bern}(d)$,
so,
$H(Y\mid X=x)\le H(\mathbf{D}\mid X=x) \le H(\mathbf{D})=n\,h_2(d)$.
Hence
\(
  I_U(X;Y)=H_U(Y)-H_U(Y\mid X)
  \ge n(1-d)\log_2 q-n\,h_2(d).
\)
Dividing by $n$ gives the lower bound.

\textit{Upper bound.}
The output length $K\triangleq |Y|\sim\mathrm{Bin}(n,1-d)$ regardless of $p_X$,
so $H(K)=H_{\mathrm{Bin}}(n,d)$.  Since $H(Y\mid K=k)\le k\log_2 q$:
\(
H(Y)=H(K)+H(Y\mid K)\le H_{\mathrm{Bin}}(n,d)+n(1-d)\log_2 q.
\)
Since deletions are independent of $X$, we have $H(Y\mid X)\ge H(K)$,
so
\(
I(X;Y)=H(Y)-H(Y\mid X)\le n(1-d)\log_2 q.
\)
This holds for any $p_X$, giving $C_{q,n}\le(1-d)\log_2 q$ and,
as $n\to\infty$, $C_q\le(1-d)\log_2 q$.
\end{proof}

\begin{theorem}[Leading-order small-$d$ lower bound]
\label{thm:smalld}
For all $q\ge 2$, as $d\to 0^+$:
\begin{equation}
  C_q(d)\;\ge\;\log_2 q+d\log_2 d+O(d).
  \label{eq:smalld}
\end{equation}
\end{theorem}

\begin{proof}
Expanding $h_2(d)=-d\log_2 d+d\log_2 e+O(d^2)$, the lower bound in
Theorem~\ref{thm:sandwich} gives
\(
C_q(d)\ge\log_2 q+d\log_2 d-d\log_2(qe)+O(d^2\log(1/d)),
\)
so $C_q(d)-\log_2 q\ge d\log_2 d+O(d)$ as claimed.
\end{proof}

\begin{remark}
The upper bound \(C_q(d)\le (1-d)\log_2 q\) is classical and follows from the
trivial erasure-channel converse; see, for example,~\cite{mitzenmacher2009,rahmati2014}.
For lower bounds, Diggavi and Grossglauser's early deletion-channel work~\cite{diggavi2001}
introduced a framework that has been widely used in later analyses.
In the $q$-ary regime, the lower bounds obtained from that line of work imply
\(
C_q(d)\ge \log_2 \tfrac{q}{q-1}+(1-d)\log_2 (q-1)-h(d)
\), which is better than \(C_q(d)\ge (1-d)\log_2 q-h(d)\).
The bound \((1-d)\log_2 q-h_2(d)\) in this paper is a refined achievable-rate
expression obtained from the uniform-input entropy calculation.
\end{remark}

\begin{remark}
The upper bound $C_q(d)\le(1-d)\log_2 q$ is classical.
In the binary case ($q=2$), the lower bound $(1-d)-h_2(d)$
dominates the constant-factor bounds $0.1185(1-d)$~\cite{drinea2007}
and $0.1221(1-d)$~\cite{rubinstein2023} only in the low-deletion
regime: specifically, $1-d-h_2(d)>0.1221(1-d)$ holds for
$d\lesssim 0.193258$.
\end{remark}

\section{Main Results 2: Tightening via the Entropy Correction $\Phi$}
\label{sec:tighter}

The lower bound in Theorem~\ref{thm:sandwich} bounds $H(Y\mid X)$ by
$n\,h_2(d)$, treating each deletion event as maximally uncertain.
But when $N_n(x,y)\ge 2$, multiple deletion subsets map the same
input to the same output~--- the conditional distribution is more
concentrated than the bound assumes.  We now quantify this
concentration through a single correction term derived from the
integer values of $N_n$.

\begin{definition}[Entropy correction]
\label{def:Phi}
For $0\le k\le n$ and $x\in\Sigma_q^n$, define the
\emph{per-input correction}
\begin{equation}
  \phi_{k,n}(x)\triangleq \frac{1}{\binom{n}{k}}
  \sum_{y\in\Sigma_q^k}N_n(x,y)\log_2 N_n(x,y)\ge 0,
  \label{eq:phi}
\end{equation}
and the \emph{average correction}
\begin{equation}
  \Phi_{k,n}\triangleq \frac{1}{q^n\binom{n}{k}}
  \sum_{x\in\Sigma_q^n}\sum_{y\in\Sigma_q^k}
  N_n(x,y)\log_2 N_n(x,y)\ge 0.
  \label{eq:Phi}
\end{equation}
Set $\Delta_n(d)\triangleq \sum_{k=1}^{n-1}w_k\Phi_{k,n}\ge 0$.
\end{definition}

Note that $\phi_{k,n}(x)=0$ if and only if every entry $N_n(x,y)
\in\{0,1\}$, which holds at the boundary levels $k=0$ and $k=n$
but not for $1\le k\le n-1$ when $n\ge 2$ (see Example~\ref{ex:phi}).

\begin{example}[$n=2$, $q=2$: correction terms]
\label{ex:phi}
From Example~\ref{ex:n2q2}, the only non-trivial level is $k=1$.
Using $\binom{2}{1}=2$:
\(
\phi_{1,2}(00)=\tfrac{1}{2}[2\log_2 2+0]=1,
\phi_{1,2}(01)=\tfrac{1}{2}[1\cdot 0+1\cdot 0]=0,
\phi_{1,2}(10)=0,
\phi_{1,2}(11)=1.
\)
$\phi_{1,2}(x)>0$ precisely for the constant strings $x\in\{00,11\}$:
deleting either symbol always yields the same output, so the length-$1$
output is more concentrated.  The average is
$\Phi_{1,2}=\frac{2\log_2 2+2\log_2 2}{4\cdot 2}=\frac{1}{2}$.
With $w_1=2d(1-d)$:
\(
\Delta_2(d)=w_1\cdot\Phi_{1,2}=d(1-d)>0\quad\forall\,d\in(0,1).
\)
\end{example}

\begin{theorem}[Tightened capacity lower bound]
\label{thm:tighter_HYX}
For any fixed $x\in\Sigma_q^n$:
\begin{equation}
  H(Y\mid X=x)= n\,h_2(d)-\sum_{k=0}^n w_k\phi_{k,n}(x).
\end{equation}
Averaging over the uniform input $X\sim\mathrm{Unif}(\Sigma_q^n)$:
\begin{equation}
  H(Y\mid X)= n\,h_2(d)-\Delta_n(d).
  \label{eq:tighter_HYX}
\end{equation}
\end{theorem}

\begin{proof}
Fix $x$.  Given $K=k$, the conditional distribution is
$\Pr(Y=y\mid X=x,K=k)=N_n(x,y)/\binom{n}{k}$, so
\begin{align}
  H(Y\mid & X=x,K=k) \nonumber \\
  &=\log_2\binom{n}{k}
   -\frac{1}{\binom{n}{k}}\sum_y N_n(x,y)\log_2 N_n(x,y) \nonumber \\
  &=\log_2\binom{n}{k}-\phi_{k,n}(x).
\end{align}
Since $K\sim\mathrm{Bin}(n,1-d)$ is independent of $X$, 
so after averaging over $K$, 
\(
H(Y \mid X = x) 
= H(K \mid X = x) + H(Y \mid X = x, K),
\)
we have $H(K \mid X = x) = H(K)$, so:
\(
H(Y \mid X = x) = H(K) + \sum_{k=0}^{n} w_k H(Y \mid X = x, K = k)
\)
which gives
\(
H(Y \mid X = x) = H(K)
+ \sum_{k=0}^{n} w_k \log_2\! \binom{n}{k}
- \sum_{k=0}^{n} w_k \,\phi_{k,n}(x),
\)
and using
the identity $H(K)+\sum_k w_k\log_2\binom{n}{k}
=n\,h_2(d)$:
\(
  H(Y\mid X=x)= n\,h_2(d)-\sum_{k=0}^n w_k\phi_{k,n}(x).
\)
Averaging over $p_X$ and using $\phi_{0,n}(x)=\phi_{n,n}(x)=0$
gives~\eqref{eq:tighter_HYX}.
\end{proof}

\begin{example}[Tightened bound with exact equality at $n=2$, $q=2$]
\label{ex:equality}
We verify that the tightened bound holds with equality under the
uniform input.  Since $K=|Y|\sim\mathrm{Bin}(2,1-d)$ is independent
of $X$, and only $k=1$ contributes non-trivially (the $k=0$ and $k=2$
outputs are deterministic given $x$):
\(
H(Y\mid X=x)=H_{\mathrm{Bin}}(2,d)+2d(1-d)[1-\Phi_{1,2}(x)].
\)
Averaging uniformly:
\(
H(Y\mid X)
  =H_{\mathrm{Bin}}(2,d)+2d(1-d)\bigl[1-\tfrac{1}{4}(1+0+0+1)\bigr]
  =H_{\mathrm{Bin}}(2,d)+d(1-d).
\)
Using $H_{\mathrm{Bin}}(2,d)=2h_2(d)-2d(1-d)$ gives 
\(
H(Y\mid X)=2h_2(d)-d(1-d)=2h_2(d)-\Delta_2(d).
\)
The gap relative to the loose bound $2h_2(d)$ is exactly
$\Delta_2(d)=d(1-d)>0$ for all $d\in(0,1)$.
\end{example}

\begin{theorem}[Exact uniform-input rate]
\label{thm:exact_uniform_rate}
Under the uniform input $X\sim\mathrm{Unif}(\Sigma_q^n)$,
\begin{align}
  \frac{1}{n}I_U(X;Y)
  = (1-d)\log_2 q
    + \frac{1}{n}H_{\mathrm{Bin}}(n,1-d) \nonumber \\
    - h_2(d)
    + \frac{\Delta_n(d)}{n}.
  \label{eq:exact_uniform_rate}
\end{align}
\end{theorem}

\begin{proof}
By Lemma~\ref{thm:HY},
\(
H(Y)=n(1-d)\log_2 q+H_{\mathrm{Bin}}(n,1-d).
\)
By Theorem~\ref{thm:tighter_HYX},
\(
H(Y\mid X)=n\,h_2(d)-\Delta_n(d).
\)
Subtracting gives~\eqref{eq:exact_uniform_rate}.
\end{proof}

\begin{corollary}[Exact uniform-input lower bound]
\label{cor:exact_uniform_lb}
For $n\ge 2$ and $d\in(0,1)$:
\begin{align}
  C_{q,n}\ge (1-d)\log_2 q
    + \frac{1}{n}H_{\mathrm{Bin}}(n,1-d)
    - h_2(d)
    + \frac{\Delta_n(d)}{n}.
  \label{eq:tighter_lb}
\end{align}
\end{corollary}

\begin{proof}
By Theorem~\ref{thm:exact_uniform_rate}, under the uniform input 
$X\sim\mathrm{Unif}(\Sigma_q^n)$, ~\eqref{eq:exact_uniform_rate}.
Since $C_{q,n}$ is the per-symbol capacity, it is at least the mutual
information achieved by any particular input distribution. Therefore,
\(
C_{q,n}\ge \frac{1}{n}I_U(X;Y),
\)
which gives~\eqref{eq:tighter_lb}.
\end{proof}

\begin{remark}[Comparison with the known lower bound~\cite{diggavi2001}]
\label{rem:B2_compare}
Let
\begin{align}
LB_2(q,d)\triangleq \log_2\frac{q}{q-1}+(1-d)\log_2(q-1)-h_2(d) \notag \\
        =\log_2 q-d\log_2(q-1)-h_2(d).
        \label{eq:B2}
\end{align}
Then the exact uniform-input rate satisfies
\(
\frac{1}{n}I_U  (X;Y)  -LB_2(q,d)
=\frac{1}{n}H_{\mathrm{Bin}}(n,1-d)+\frac{\Delta_n(d)}{n}
-d\log_2\frac{q}{q-1}.
\)
Hence, for every fixed $q\ge 2$ and $n\ge 2$, the exact uniform-input
rate exceeds $LB_2(q,d)$ for all sufficiently small $d>0$.
\end{remark}

\begin{lemma}[Recursive lower bound on $\Phi_{k,n}$]
\label{lem:Phi_rec}
For $1\le k\le n-1$:
\begin{equation}
  \Phi_{k,n}\ge\frac{n-k}{n}\Phi_{k,n-1}+\frac{k}{n}\Phi_{k-1,n-1}.
  \label{eq:Phi_rec}
\end{equation}
\end{lemma}

\begin{proof}
Define the auxiliary quantity
\(
  S_{k,n}\triangleq\sum_{x\in\Sigma_q^n}\sum_{y\in\Sigma_q^k}
  N_n(x,y)\log_2 N_n(x,y),
\)
so that $\Phi_{k,n}=S_{k,n}/(q^n\binom{n}{k})$ by
Definition~\ref{def:Phi}.  It suffices to show
\begin{equation}
  S_{k,n}\ge q\,S_{k,n-1}+q\,S_{k-1,n-1}.
  \label{eq:Srec}
\end{equation}

\paragraph{Super-additivity of $f(u)=u\log_2 u$.}
We claim $f(a+b)\ge f(a)+f(b)$ for all $a,b\ge 0$.  The cases
$a=0$ or $b=0$ are trivial since $f(0)=0$.  For $a,b>0$:
\(
  f(a+b)-f(a)-f(b)
  =(a+b)\log_2(a+b)-a\log_2 a-b\log_2 b
  =a\log_2\!\Bigl(1+\frac{b}{a}\Bigr)
   +b\log_2\!\Bigl(1+\frac{a}{b}\Bigr)\;\ge\;0,
\)
since both terms are non-negative for $a,b>0$.

\paragraph{Applying super-additivity entrywise.}
From the recursion~\eqref{eq:rec},
for $x=sx'\in\Sigma_q^n$ and $y=ty'\in\Sigma_q^k$.
Super-additivity of $f$ gives, for every $s,t,x',y'$:
\begin{align}
  &f\bigl(N_n(sx',ty')\bigr) \nonumber\\
  &\ge f\bigl(N_{n-1}(x',ty')\bigr)
   +\mathbf{1}_{\{s=t\}}f\bigl(N_{n-1}(x',y')\bigr).
  \label{eq:fineq}
\end{align}

\paragraph{Summing to obtain~\eqref{eq:Srec}.}
Sum~\eqref{eq:fineq} over all $s,t\in\Sigma_q$,
$x'\in\Sigma_q^{n-1}$, $y'\in\Sigma_q^{k-1}$.  The left-hand side
gives $S_{k,n}$.  On the right-hand side:
\begin{itemize}
  \item \emph{First term.} Summing $f(N_{n-1}(x',ty'))$ over
        $s\in\Sigma_q$ contributes a factor of $q$ (since the summand
        is independent of $s$), and then summing over $t\in\Sigma_q$,
        $x'\in\Sigma_q^{n-1}$, $y'\in\Sigma_q^{k-1}$ gives
        $q\,S_{k,n-1}$.
  \item \emph{Second term.} The indicator $\mathbf{1}_{\{s=t\}}$
        contributes exactly one value of $s$ per $t$.  Summing over
        $t\in\Sigma_q$, $x'\in\Sigma_q^{n-1}$, $y'\in\Sigma_q^{k-1}$
        then gives $q\,S_{k-1,n-1}$.
\end{itemize}
Hence $S_{k,n}\ge q\,S_{k,n-1}+q\,S_{k-1,n-1}$.

\paragraph{Converting back to $\Phi$.}
Dividing both sides of~\eqref{eq:Srec} by $q^n\binom{n}{k}$:
\(
  \Phi_{k,n}
  \ge \frac{q\,S_{k,n-1}}{q^n\binom{n}{k}}
      +\frac{q\,S_{k-1,n-1}}{q^n\binom{n}{k}}
  = \frac{q^{n-1}\binom{n-1}{k}\,\Phi_{k,n-1}}{q^{n-1}\binom{n}{k}}
    +\frac{q^{n-1}\binom{n-1}{k-1}\,\Phi_{k-1,n-1}}{q^{n-1}\binom{n}{k}}
  = \frac{\binom{n-1}{k}}{\binom{n}{k}}\,\Phi_{k,n-1}
    +\frac{\binom{n-1}{k-1}}{\binom{n}{k}}\,\Phi_{k-1,n-1}
  = \frac{n-k}{n}\,\Phi_{k,n-1}+\frac{k}{n}\,\Phi_{k-1,n-1},
\)
where we used $S_{k,n-1}=q^{n-1}\binom{n-1}{k}\,\Phi_{k,n-1}$,
$S_{k-1,n-1}=q^{n-1}\binom{n-1}{k-1}\,\Phi_{k-1,n-1}$, and the
binomial ratios $\binom{n-1}{k}/\binom{n}{k}=(n-k)/n$,
$\binom{n-1}{k-1}/\binom{n}{k}=k/n$.
\end{proof}

\begin{theorem}[Closed-form lower bound for $\Phi_{k,n}$]
\label{thm:Phi_closed_form}
Assume $\Phi_{0,n}=0$ and $\Phi_{n,n}=0$ for all $n\ge 1$, and
\(
\Phi_{1,2}=\frac{1}{q}.
\)
Then for every $q\ge 2$, every $n\ge 2$, and every $1\le k\le n-1$,
\begin{equation}
  \Phi_{k,n}\ge \frac{2}{q}\,\frac{k(n-k)}{n(n-1)}.
  \label{eq:Phi_closed_form}
\end{equation}
\end{theorem}

\begin{proof}
Define
\(
\Psi_{k,n}\triangleq\binom{n}{k}\Phi_{k,n}.
\)
Using Lemma~\ref{lem:Phi_rec},
\begin{align}
\Psi_{k,n}
 =\binom{n}{k}\Phi_{k,n}
\ge &\binom{n-1}{k}\Phi_{k,n-1}+\binom{n-1}{k-1}\Phi_{k-1,n-1} \nonumber \\
&=\Psi_{k,n-1}+\Psi_{k-1,n-1}.
\end{align}
The boundary values become
\(
\Psi_{0,n}=0,\, \Psi_{n,n}=0,\, \Psi_{1,2}=\binom{2}{1}\Phi_{1,2}=\frac{2}{q}.
\)
We claim that
\(
\Psi_{k,n}\ge \frac{2}{q}\binom{n-2}{k-1},
\, (1\le k\le n-1).
\)
This is true at $(n,k)=(2,1)$, since both sides equal $2/q$.
Now assume it holds for level $n-1$. Then, for $1\le k\le n-1$,
\(
\Psi_{k,n}
\ge \Psi_{k,n-1}+\Psi_{k-1,n-1} 
\ge \frac{2}{q}\binom{n-3}{k-1}+\frac{2}{q}\binom{n-3}{k-2}
= \frac{2}{q}\binom{n-2}{k-1},
\)
To show that, set $F_{k,n}\triangleq\frac{q}{2}\Psi_{k,n}$, so $F_{k,n}\ge F_{k,n-1}+F_{k-1,n-1}$
and $F_{1,2}=1$. Since $B_{k,n}\triangleq \binom{n-2}{k-1}$ satisfies Pascal's
identity with $B_{1,2}=1$, induction gives $F_{k,n}\ge B_{k,n}=\binom{n-2}{k-1}$,
and hence $\Psi_{k,n}\ge\frac{2}{q}\binom{n-2}{k-1}$. Therefore,
\(
\Phi_{k,n}=\frac{\Psi_{k,n}}{\binom{n}{k}}\ge\frac{2}{q}\,\frac{\binom{n-2}{k-1}}{\binom{n}{k}}=\frac{2}{q}\,\frac{k(n-k)}{n(n-1)}.
\)
\end{proof}

\begin{lemma}[Lower bound on $\Delta_n(d)$]
\label{lem:Delta_bound}
For every $q\ge 2$, every $n\ge 2$, and every $d\in(0,1)$,
\begin{equation}
  \Delta_n(d)\ge \frac{2d(1-d)}{q}.
  \label{eq:Delta_bound}
\end{equation}
\end{lemma}

\begin{proof}
By definition,
\(
\Delta_n(d)=\sum_{k=1}^{n-1} \binom{n}{k} d^{n-k}(1-d)^k \,\Phi_{k,n}.
\)
Using Theorem~\ref{thm:Phi_closed_form},
\[
\Delta_n(d)\ge
\frac{2}{q\,n(n-1)}
\sum_{k=1}^{n-1}\binom{n}{k} d^{n-k}(1-d)^k\,k(n-k).
\]
Let $K\sim\mathrm{Bin}(n,1-d)$. Then the sum is exactly
\(
\mathbb{E}[K(n-K)]
= n\,\mathbb{E}[K]-\mathbb{E}[K^2]
= n(n-1)(1-d)d.
\)
Therefore,
\(
\Delta_n(d)\ge
\frac{2}{q\,n(n-1)}\,n(n-1)(1-d)d
=\frac{2d(1-d)}{q}.
\)
\end{proof}


\begin{table*}[ht]
\centering
\caption{Capacity bounds $C_{q,n}$ for $q\in\{2,3\}$, $n\in\{3,5,10\}$, $d\in\{0.05,0.10,0.20\}$. \(C_{(q,n)}\) computed by the Blahut--Arimoto algorithm \cite{blahut1972, arimoto1972}}
\label{tab:compare}
\begin{tabular}{ccc cccc cc}
\toprule
$q$ & $n$ & $d$ & $LB_1$ (this paper) & $LB_2$~\cite{diggavi2001} & KM~\cite{kanoria2013} & LB$^+$ (this paper) & $C_{(q,n)}$ & UB~\cite{mitzenmacher2009} \\
\midrule
2 & 3  & 0.05 & 0.664 & 0.714 & 0.730 & 0.910 & 0.910 & 0.950 \\
2 & 3  & 0.10 & 0.431 & 0.531 & 0.569 & 0.825 & 0.827 & 0.900 \\
2 & 3  & 0.20 & 0.078 & 0.278 & 0.372 & 0.668 & 0.676 & 0.800 \\
\midrule
2 & 5  & 0.05 & 0.664 & 0.714 & 0.730 & 0.886 & 0.887 & 0.950 \\
2 & 5  & 0.10 & 0.431 & 0.531 & 0.569 & 0.782 & 0.786 & 0.900 \\
2 & 5  & 0.20 & 0.078 & 0.278 & 0.372 & 0.602 & 0.613 & 0.800 \\
\midrule
2 & 10 & 0.05 & 0.664 & 0.714 & 0.730 & 0.851 & 0.852 & 0.950 \\
2 & 10 & 0.10 & 0.431 & 0.531 & 0.569 & 0.722 & 0.728 & 0.900 \\
2 & 10 & 0.20 & 0.078 & 0.278 & 0.372 & 0.516 & 0.531 & 0.800 \\
\midrule
3 & 3  & 0.05 & 1.219 & 1.249 & ---   & 1.453 & 1.454 & 1.506 \\
3 & 3  & 0.10 & 0.957 & 1.016 & ---   & 1.328 & 1.329 & 1.426 \\
3 & 3  & 0.20 & 0.546 & 0.663 & ---   & 1.095 & 1.101 & 1.268 \\
\midrule
3 & 5  & 0.05 & 1.219 & 1.249 & ---   & 1.425 & 1.426 & 1.506 \\
3 & 5  & 0.10 & 0.957 & 1.016 & ---   & 1.276 & 1.279 & 1.426 \\
3 & 5  & 0.20 & 0.546 & 0.663 & ---   & 1.011 & 1.020 & 1.268 \\
\midrule
3 & 10 & 0.05 & 1.219 & 1.249 & ---   & 1.386 & 1.387 & 1.506 \\
3 & 10 & 0.10 & 0.957 & 1.016 & ---   & 1.208 & 1.212 & 1.426 \\
3 & 10 & 0.20 & 0.546 & 0.663 & ---   & 0.908 & 0.920 & 1.268 \\
\bottomrule
\end{tabular}
\end{table*}

\begin{lemma}[Explicit positivity certificate]
\label{lem:positivity}
For all $q\ge 2$, $n\ge 2$, and $d\in(0,1)$:
\begin{equation}
  \Delta_n(d)\;\ge\;
  2d(1-d)^{n-1}\!\left(1-\!\left(\tfrac{q-1}{q}\right)^{n-1}\right)>0.
  \label{eq:cert}
\end{equation}
\end{lemma}

\begin{proof}
Retain only the $k=n-1$ term in $\Delta_n(d)$:
$\Delta_n(d)\ge w_{n-1}\Phi_{n-1,n}=nd(1-d)^{n-1}\Phi_{n-1,n}$.
For $k=n-1$ (exactly one deletion), $N_n(x,y)$ counts positions whose
deletion yields $y$.  If $x$ has at least one pair of equal adjacent
symbols, then two positions yield the same output, so $N_n(x,y^*)\ge 2$
for some $y^*$.  Since $f(u)=u\log_2 u$ satisfies $f(2)=2>0$, each
such string contributes at least $2$ to
$\sum_y N_n(x,y)\log_2 N_n(x,y)$.
The number of strings in $\Sigma_q^n$ with \emph{no} equal adjacent
symbols is $q(q-1)^{n-1}$, so at least $q^n-q(q-1)^{n-1}$ strings
each contribute at least $2$:
\(
\sum_{x,y}N_n(x,y)\log_2 N_n(x,y)\ge 2\bigl(q^n-q(q-1)^{n-1}\bigr).
\)
Dividing by $q^n\binom{n}{n-1}=q^n n$ gives
$\Phi_{n-1,n}\ge\frac{2}{n}(1-(({q-1})/{q})^{n-1})$.
Substituting $w_{n-1}=nd(1-d)^{n-1}$ yields~\eqref{eq:cert}.
Positivity holds since $(q-1)^{n-1}<q^{n-1}$ for $q\ge 2$, $n\ge 2$.
\end{proof}

Combining Theorem~\ref{thm:sandwich}, Corollary~\ref{cor:exact_uniform_lb},
and Lemma~\ref{lem:positivity}, we obtain
\begin{align}
  (1-d)\log_2 q &
  + \frac{1}{n}H_{\mathrm{Bin}}(n,1-d)
  - h_2(d)
  + \frac{\Delta_n(d)}{n} \nonumber \\
  &\;\le\; C_{q,n}\;\le\;
  (1-d)\log_2 q,
  \label{eq:full_sandwich1}
\end{align}
where \(\Delta\) given from \eqref{eq:cert}.
From Example~\ref{ex:n2q2} and Example~\ref{ex:phi}, we have
\(\Delta_2(d)=d(1-d)\). Thus, for \(q=2\) and \(n=2\), we have:
\(
(1-d)+\frac{1}{2}H_{\mathrm{Bin}}(2,1-d)-h_2(d)+\frac{d(1-d)}{2}
\;\le\; C_{2,2}\;\le\;(1-d).
\)

\section{Numerical Verification}
\label{sec:numerical}

Table~\ref{tab:compare} reports six bounds on the capacity per symbol
$C_{q,n}$ for $q=2,3$ and block lengths $n=3,5,10$ across three
deletion probabilities $d\in\{0.05,0.10,0.20\}$. The bounds are defined as follows.
\emph{ The basic lower bound (\(LB_1\)), $LB_1=(1-d)\log_2 q-h_2(d)$ from this paper, ~\eqref{eq:sandwich}}, Diggavi--Grossglauser bound ($LB_2$)~\cite{diggavi2001}, and Upper Bound (UB), $UB=(1-d)\log_2 q$ from ~\cite{mitzenmacher2009},
and the Diggavi--Grossglauser lower bound~\cite{diggavi2001} is ~\eqref{eq:B2}.
For $q=2$ only, the Kanoria--Montanari asymptotic expansion~\cite{kanoria2013} gives as KM, where
\(
  \mathrm{KM} = 1 + d\log_2 d - A_1\,d + A_2\,d^2,
\)
with values of $A_1=1.15416,A_2=1.67815$. Dashes indicate KM does not apply for $q>2$.
\emph{The tightened lower bound (LB$^+$) from this paper is the exact uniform-input
rate and a proved lower bound on $C_{q,n}$ as ~\eqref{eq:tighter_lb}}
\(
  \mathrm{LB}^{+} = (1-d)\log_2 q
    + \frac{1}{n}H_{\mathrm{Bin}}(n,1-d)
    - h_2(d)
    + \frac{\Delta_n(d)}{n}.
\)
The exact finite-block capacity per channel use of the $q$-ary with blocklength $n$ is denoted by $BAC_{(q,n)}$ and is computed using the
Blahut--Arimoto algorithm.
The ordering
\(
  \mathrm{LB_1} \;\le\; LB_2 \;\le\; \mathrm{LB}^{+} \;\le\; C_{q,n} \;\le\; \mathrm{UB}
\)
holds in every entry of the table, confirming the capacity sandwich.
The Blahut--Arimoto values lie strictly above $\mathrm{LB}^{+}$ in all
cases, and the gap narrows as $n$ grows.

\section{Conclusion}
\label{sec:conclusion}

The pattern-count scalar $N_n(x,y)$ — the number of distinct deletion
subsets mapping $x$ to $y$ — provides a self-contained combinatorial
foundation for the $q$-ary deletion channel.  Its two-case recursion
holds uniformly for all $q\ge 2$.  The row-sum identity
$\sum_y N_n=\binom{n}{k}$ certifies stochastic normalization via the
binomial theorem; the column-sum identity $\sum_x N_n=q^{n-k}\binom{n}{k}$,
propagated through Pascal's identity, implies uniform output under
uniform input.
These properties chain together cleanly: uniform output
$\Rightarrow$ exact entropy formula~\eqref{eq:HY}
$\Rightarrow$ capacity sandwich~\eqref{eq:sandwich}
$\Rightarrow$ tightened bound~\eqref{eq:tighter_lb} via the correction
$\Phi_{k,n}$ extracted from the integer values of $N_n$
$\Rightarrow$ small-$d$ lower bound~\eqref{eq:smalld}.
The positivity of $\Delta_n(d)$ is certified explicitly by the
count of strings with repeated adjacent symbols.
Future directions include tight asymptotic characterization of
$\Delta_n(d)/n$ as $n\to\infty$, code design exploiting the recursive
integer structure of $N_n$, and extension to insertion-deletion channels.


\end{document}